\documentclass[
final
]{dmtcs-episciences}

\usepackage[utf8]{inputenc}
\usepackage{subfigure}

\usepackage[round]{natbib}

\usepackage{amssymb}

\usepackage{float}
\usepackage{mathtools}
\DeclarePairedDelimiter\ceil{\lceil}{\rceil}
\DeclarePairedDelimiter\floor{\lfloor}{\rfloor}

\newtheorem{theorem}{Theorem}
\newtheorem{fact}{Fact}

\usepackage{algorithm} 
\usepackage{algpseudocode}
\newcommand{\Break}{\State \textbf{break} }

\usepackage{tikz}
\usetikzlibrary{automata,positioning}

\author{U\u{g}ur K\"{u}\c{c}\"{u}k\affiliationmark{1}
  \and A. C. Cem Say\affiliationmark{1}
  \and Abuzer Yakary{\i}lmaz\affiliationmark{2,3}\thanks{Yakary{\i}lmaz was
partially supported by CAPES with grant 88881.030338/2013-01 and ERC
Advanced Grant MQC.}}
\title{Inkdots as advice for finite automata}

\affiliation{
  Bo\u{g}azi\c{c}i University, Istanbul, Turkey\\
  National Laboratory for Scientific Computing, Petr\'{o}polis, RJ, Brazil\\
  University of Latvia, Faculty of Computing, R\={\i}ga, Latvia}


\received{2015-9-16}
\revised{2016-12-5}
\accepted{2017-9-1}

\begin{document}

\publicationdetails{19}{2017}{3}{1}{1286}

\maketitle

\begin{abstract}
We examine inkdots placed on the input string as a way of providing advice to  finite automata, and establish the relations between this model and the previously studied models of advised finite automata. The existence of an infinite hierarchy of classes of languages that can be recognized with the help of increasing numbers of inkdots as advice is shown. The effects of different forms of advice on the succinctness of the advised machines are examined. We also study randomly placed inkdots as advice to probabilistic finite automata, and demonstrate the superiority of this model over its deterministic version. Even very slowly growing amounts of space can become a resource of meaningful use if the underlying advised model is extended with access to secondary memory, while it is famously known that such small amounts of space are not useful for unadvised one-way Turing machines.
\end{abstract}

\section{Introduction}
\label{sec:in}

Advice is a piece of trusted supplemental information that is provided to a computing machine in advance of its execution in order to assist it in its task (see (\cite{KL82})). Typically, the advice given to a machine depends only on the length, and not the full content of its input string. Advised computation has been studied on different models and in relation with various concepts. A variety of methods have also been proposed for providing such supplemental information to finite automata, and their power and limitations have been studied extensively (see \textit{e.g.} (\cite{DH95}), (\cite{TYL10}) and (\cite{KSY14})). 

In (\cite{DH95}),   Damm and Holzer examined the first advised finite automaton model, in which  the advice is prepended to the original input prior to the computation of the automaton. In this \emph{advice prefix} model, the machine starts by scanning the advice, which  has the effect of selecting a  state at which to start the  processing of the actual input string. Letting $\mathsf{REG/}k$ denote the class of languages that can be recognized by deterministic finite automata with binary prefix advice of constant length $k$, it is known that $\mathsf{REG/}k\subsetneq\mathsf{REG/}(k+1)$, for all $k\geq 0$. However, allowing the advice length to be an increasing function of the input length  does not bring in any more power, since a finite automaton with $s$ states which takes a long advice prefix can always be imitated by another automaton which takes an advice prefix of length $\ceil{\log s}$.

In the \emph{advice track} model, the advice is placed on a separate track of the read-only input tape, so that it can be scanned in parallel with the actual input by the single tape head. This model was first examined by Tadaki et al. in (\cite{TYL10}), and later by T. Yamakami in a series of papers (\cite{Ya08,Ya10,Ya11,Ya14}). Unlike the advice prefix model,  it is possible to utilize advice strings whose lengths grow linearly in terms of input length meaningfully in this setup. As a reflection, the class of languages recognizable by deterministic finite automata with an advice track is denoted by $\mathsf{REG/}n$ and is different from $\mathsf{REG/}k$. The language $\{a^mb^m \mid m>0\}$, for instance, can be shown easily to be in $\mathsf{REG/}n$, whereas it is not contained in  $\mathsf{REG/}k$ for any $k\geq 0$.   Probabilistic and quantum variants of this model were also studied, and many relations have been established among the classes of languages that can be recognized in each setting. In the sequel, we will sometimes need to specify the size of the alphabet supporting the strings on the advice track. In such cases, the number of elements in the advice alphabet will also appear in the class name, \textit{e.g.} $\mathsf{REG/}n(2)$ is the subset of $\mathsf{REG/}n$ where advice strings are restricted to be on a binary alphabet.

Another model of advised finite automata, which incorporates one or more separate tapes for the advice, was introduced by Freivalds in (\cite{Fr10,AF10}).  In this model, the automaton has two-way access to both the input and the advice tapes. Unlike the advice prefix and advice track models, Freivalds' model requires the advice string for inputs of length $n$ to be helpful for all shorter inputs as well.

In (\cite{KSY14}),  K\"{u}\c{c}\"{u}k et al. studied yet another model of finite automata with advice tapes, where the advice string is provided on a separate read-only tape, the access to which is limited to be one-way. The content of the advice depends only on the length of the input.  Deterministic, probabilistic, and quantum variants of the model under several settings for input access were examined, and separation results were obtained for the classes of languages that can be recognized under certain resource bounds in many of these cases.

In this paper, we introduce a new way of providing advice to finite automata, namely, by marking some positions on the input with inkdots. 
This model is presented in detail in Section \ref{sec:model}, where  it is shown to be intermediate in power between the prefix and track models of advice. Section \ref{sec:moredots} demonstrates the existence of an infinite hierarchy among the classes of languages that can be recognized with the help of different numbers of inkdots as advice; both when those numbers are restricted to be constants, and when bounded by  increasing functions of the input length. In Section \ref{sec:succinctness}, we show that inkdot advice can cause significant savings in the number of states of the advised automaton when compared with the prefix advice model, which in turn is superior in this sense to pure unadvised automata. 

In Section \ref{sec:random}, we demonstrate that the strength of the model increases if one employs a probabilistic automaton instead of a deterministic one, and assists it with inkdots placed randomly according to an advice distribution. Section \ref{sec:sublogspace} extends the advised machine model by allowing it access to a work tape. It is interesting to note that arbitrarily small space turns out to be a fruitful computational resource along with advice, while it is well known that one-way sublogarithmic-space Turing machines (TM's) cannot recognize more languages than their constant-space counterparts.  Section \ref{sec:conc}  concludes the paper with a list of open questions.

\section{Inkdots as advice to finite automata}
\label{sec:model}

In this section, we  introduce a new model of advised finite automata in which the advice will be provided as inkdots on the input string. An inkdot is a supplementary marker that can be placed on any symbol of the input  of a computing machine. The presence of that mark can be sensed  by the automaton only when the input tape head visits that position. This mark can not be erased, and no more than one inkdot is allowed on one cell.  (Inkdots are different from pebbles, which are more widely used in the literature (see \textit{e.g.} (\cite{Sz94})), only in their inability to be moved around on the input tape.) It is known (see (\cite{RCH91})) that a deterministic Turing machine would not gain any additional computational power if it is also provided with the ability of marking one input tape cell  with an inkdot. 

A finite automaton that takes inkdots as advice does not have the power to mark cells on the input tape with inkdots, however, it can sense these marks if they are present on the currently scanned input cell.  The inkdots are assumed to be placed prior to the execution of the machine in accordance with an advice function which maps the length of the input string to a set of positions on the input string where the inkdots are to be placed. A deterministic finite automaton (dfa) with inkdot advice can then be defined in precisely the same way as a standard unadvised dfa (see \textit{e.g.} (\cite{Si06})), but with an extended input alphabet containing the       ``dotted," as well as original, versions of the symbols of the actual input alphabet.     

  The class of languages that can be recognized with the help of advice consisting of a certain amount $a$ of inkdots by a dfa will be denoted $\mathsf{REG}/a(\odot)$. In this expression, $a$ may either be a natural number, describing cases where the advised machine is supposed to use at most that constant number of inkdots regardless of the length of the input, or a function $f(n)$, indicating that at most $f(n)$ inkdots can appear in the advice for inputs of length $n$. Note that $f(n)\leq n$ in any case.

Later in the paper, we will consider machines that  can pause the  input head in some computational steps, rather than move it to the right in each step as required for the finite automata studied in most of the manuscript.
Notation for representing language classes associated with this access mode, as well as for non-constant memory and randomized advice,  will be introduced as we examine those issues in subsequent sections.


\subsection{Inkdots vs. prefix strings as advice}
We start by establishing that the inkdot advice model is stronger than the prefix model, even when it is restricted to a single inkdot per advice.
\begin{theorem}
\label{theorem:inkdot_prefix}
For every  $k \in \mathbb{N}$, $ \mathsf{REG}/k  \subsetneq \mathsf{REG}/1(\odot)$. 
\end{theorem}

\begin{proof}
We first show that  $ \mathsf{REG}/k  \subseteq \mathsf{REG}/1(\odot)$ for all  $k \in \mathbb{N}$. 
Let $L$ be a language that is recognized by a dfa $M$ using $k$ bits of binary prefix advice. Without loss of generality, assume that $L$ is defined on a binary alphabet. One can construct a finite automaton $N$ that recognizes $L$ with the help of a single inkdot as advice as follows.

$N$ will use a lookup table to treat inputs shorter than $2^k$ bits on its own, without utilizing advice. For longer inputs, view each advice string given to $M$ as a natural number $b$ written with $k$ bits. This information will be conveyed  by placing the inkdot on the $(b+1)$'st symbol of the input to $N$.

Note that reading the prefix advice may bring $M$  into one of at most of $2^k$ different states when it is positioned at the beginning of the actual input. $N$ simulates at most $2^k$  different instances of $M$ starting from each of those different initial states in parallel on the actual input. When it scans the inkdot, $N$ uses the information encoded in the inkdot's position to pick one of the simulated machines, run it to its completion, and report its outcome as the result of its computation.

Having proven the subset relation, it remains to exhibit a language  recognized by a finite automaton that takes an inkdot as advice but not by any dfa that takes prefix advice. 
Consider the language $\{a^mb^m \mid m \in \mathbb{N}\}$, which can not be recognized by any finite automaton with advice prefix,  by Propositions 1 and 7 of  (\cite{DH95}). An inkdot marking the $(n/2)+1$'st symbol for inputs of even length $n$ is sufficient to help a dfa recognize this language, since the machine need only check that the string is of the form $a^*b^*$, and that the first $b$ appears precisely at the marked position.
\end{proof}


\subsection{Inkdots vs. the advice track}

It is evident that inkdot advice is a special case of the advice track model, when the advice alphabet is restricted to be binary. Recall that $\mathsf{REG}/n(t)$ denotes the class of languages recognized by dfa's with the aid of  advice written on a track using a $t\mbox{-ary}$ alphabet.

\begin{theorem}\label{theorem:inkdot_track}
$ \mathsf{REG}/n  (\odot)  = \mathsf{REG}/n(2)$. 
\end{theorem}
\begin{proof}
Any inkdot pattern on the $n$-symbol-long input string corresponds to a unique $n$-bit advice string, with (say) $1$'s for the marked positions, and 0's for the rest. The way the advice track is accessed simultaneously with the input track makes the two simulations required for proving the equality trivial.

\end{proof}

The reader is thus referred to (\cite{TYL10}) and (\cite{Ya10}) for what is already known about the capabilities and limitations of  advice read from tracks with binary alphabets. In particular, Theorem 2 of (\cite{Ya10}), which is cited as Fact \ref{yamakami_characterization} below, provides a straightforward characterization of  $\mathsf{REG}/n$, and can thus be used to show that certain languages are not  in  $\mathsf{REG}/n$, and therefore neither in $ \mathsf{REG}/n  (\odot)$.

\begin{fact}(\cite{Ya10})\label{yamakami_characterization}
For any language S over an alphabet $\Sigma$, the following two statements are equivalent. Let $\Delta = \{ (x,n) \in \Sigma^* \times  \mathbb{N} \mid |x| \leq n  \}$.

\begin{enumerate}
\item  $S$ is in $\mathsf{REG}/n$.
\item There is an equivalence relation $\equiv_S$ over $\Delta$  such that
\begin{enumerate}
\item the total number of equivalence classes in $\Delta / \equiv_S$  is finite, and
\item for any length $n\in \mathbb{N} $ and any two strings $x, y \in \Sigma^* $ with $|x| = |y| \leq n$, the
following holds: $(x,n) \equiv_S (y, n)$ iff,  for all $z$ with $|xz| = n$, $xz \in S  \Leftrightarrow  yz \in S$.
\end{enumerate}

\end{enumerate}

\end{fact}

In order to show the difference between the inkdot model and the general advice track approach, we consider bigger track alphabets.

\begin{theorem}
\label{theorem:inkdot_track_with_larger_alphabet}
  $\mathsf{REG}/n(2) \subsetneq \mathsf{REG}/n(k)$ for all $k\in Z$ with $k > 2$.
\end{theorem}

\begin{proof}
Since it  is trivially true that $\mathsf{REG}/n(2) \subseteq \mathsf{REG}/n(k)$ for all $k\in Z$ with $k > 2$, we just need to show the existence of languages that can be recognized by finite automata with the help of  $k\mbox{-}$ary advice but not with the help of binary advice supplied on the advice track.

Let $w = w_1w_2\cdots$ be an infinite (Martin-L\"{o}f) random  binary sequence. Let  $w_{i}$ denote $i$'th symbol of $w$ and let $w_{i:j}$ denote the segment of $w$ which starts with $w_i$ and ends with $w_j$, where $i \leq j$.  As $w$ is random, none of its prefixes can be compressed by more than an additive constant, \textit{i.e.} there is a constant $c$ such that $K(w_{1:n}) \geq n -  c $ for all $n$ where $K(w_{1:n})$ stands for the Kolmogorov complexity (or equivalently, size of the minimum description) of $w_{1:n}$.  

With reference to  $w$, we will define a $k\mbox{-ary}$ (for arbitrary $k>2$) language $\mathtt{L}_w$,   which has exactly one member, $l_i$, of each length $i \in Z^+$. For $i \in Z^+$, we divide $w$ into consecutive subwords $s_i$ the lengths, $|s_i|$ of which will be specified further below. We obtain each member $l_i$ of  $\mathtt{L}_w$  from the corresponding subword $s_i$ of $w$ by first reading $s_i$ as a binary number, then converting it to a $k\mbox{-ary}$ number, and then padding with zeros to the left if necessary to make the result $i$ symbols long. This relation can be expressed by a function $f$ as $l_i =f(s_i , k, i)$. We are yet to specify the lengths of each subword $s_i$ in $w$. We need to be able to encode the content of $s_i$  into $i$ $k\mbox{-ary}$ symbols. This entails  $|s_i| \leq \floor{ log_2( {k^i} )}$. So we set $|s_i| = \floor{ log_2( {k^i} )}$ for maximum compression. Then we can formally define  $\mathtt{L}_w$ as

\begin{displaymath}
\mathtt{L}_w = \{ f(w_{a:b} , k,i) \mid   a = \sum_{u=1}^{i-1}( {\floor{ log_2( {k^u} )}}) + 1 , \  b = \sum_{u=1}^{i}( {\floor{ log_2( {k^u} )}}) \ \mbox{for}\ i \in Z^+ \}.
\end{displaymath}

Since  $\mathtt{L}_w$ has exactly one member of each length, it is obvious that providing this member as advice and letting the automaton  check the equality of the advice and the input would suffice to recognize  $\mathtt{L}_w $ when a $k\mbox{-ary}$ alphabet is used on the advice track. Therefore, we have  $\mathtt{L}_w \in \mathsf{REG}/n(k)$. 

Now let us assume  $\mathtt{L}_w  \in \mathsf{REG}/n(2)$, which would mean that there is a dfa $M$ which recognizes  $\mathtt{L}_w $ with the help of binary advice on the advice track. Let  $a_1, a_2,\ldots, a_n$ be the binary advice strings to $M$ for inputs of length $1,2,\ldots,n$ respectively. Then Algorithm \ref{Program_P} working on the input $a_1, a_2,\ldots, a_n$ would compute and print the concatenation of the  first $n$ subwords $s_1,\ldots,s_n$ of $w$.

\begin{algorithm}
\caption{A short program for printing a prefix of $w$}
\label{Program_P}
\begin{algorithmic}[1]
 \State on input  $a_1, a_2,\ldots, a_n$ 
 \State result = blank
\For{$ i$ = 1 to $n$ }
	\For{every $k\mbox{-ary}$ string $s$ of length $i$ }
	\State Simulate $M$ on input $s$ with advice $a_i$
	\If {$M$ accepts $s$}
		\State $s_i$ = \textsc{TranslateToBinary}($s$, $\floor{ log_2( {k^i} )} $)
		\State result = concat(result, $s_i$)
		\Break 
	\EndIf
	\EndFor
\EndFor
\State print result
\State end
\Procedure{TranslateToBinary}{word, length}
	\State binary-number = transform word from   $k\mbox{-}$ary to binary
	\State binary-word = pad binary-number with zeros to the left until it fits length
	\State return binary-word 
\EndProcedure
\end{algorithmic}
\end{algorithm}

So Algorithm \ref{Program_P} (which is, say, $c$ bits long) and the first $n$ advice strings for $M$ form a description of the prefix of $w$ formed by concatenating the  first $n$ subwords $s_1,\ldots,s_n$. But this violates the incompressibility of $w$, since for large values of $n$, the total length of $s_1,\ldots,s_n$ is given by $\sum_{u=1}^{n}( {\floor{\log_2k^u}})$, which will  be approximately $\log_2{k}$ times the length of its description,  $ c + \sum_{u=1}^{n}( u)$.

Therefore, we conclude that $L\notin \mathsf{REG}/n(2)$.
\end{proof}

\section{Language recognition with varying amounts of inkdots}
\label{sec:moredots}

Having shown that a single inkdot is a more powerful kind of advice to dfa's than arbitrarily long prefix strings, and that 
no combination of inkdots could match the full power of the track advice model, we now examine the finer structure of the classes of languages recognized by dfa's advised by inkdots. A natural question to ask is: How does the recognition power increase when one allows more and more inkdots as advice? Or equivalently: How stronger does a finite automaton get when one allows more and more $1$'s in the binary advice string written on its advice track? 

It is easy  to see that no more than $\floor{n/2}$ inkdots need be used on any input of length $n$: Any dfa that uses advice which contains more than this amount of inkdots
 can be imitated by another dfa that takes a ``flipped" advice obtained by swapping the marked and unmarked positions on the input, and performs the same computation.
   
We start by establishing the existence of an infinite hierarchy in the class of languages recognized by dfa's aided with a constant number of inkdots, independent of the input length.
It will be shown that  $m +1$ inkdots are stronger than $m$ inkdots as advice to dfa's for any value of $m$. The following 
family of languages, defined on the binary alphabet $\{0,1\}$, will be used in our proof:

$\mbox{For}\  m \in \mathbb{Z^+}, \ \mathtt{L}_m= \{  (0^i1^i)^{\ceil{m/2}}(0^i)^{ m+1-2\ceil{m/2}} ~\vert ~ i>0 \}. $ In other words, $\mathtt{L}_m$ is the set of strings that consist of $m+1$ alternating equal-length segments of 0's and 1's. $\mathtt{L}_1$, for example, is   the well known language $\{ 0^n1^n \mid n>0 \}$, and $\mathtt{L}_2 = \{ 0^n1^n0^n \mid n>0 \}$.

\begin{theorem}\label{theorem:inkdot_hierarchy}
For every  $m \in \mathbb{Z^+}$,  $ \mathsf{REG}/(m-1)  (\odot)  \subsetneq \mathsf{REG}/ m (\odot) $. 
\end{theorem}
\begin{proof}
We start by  showing that $m$ inkdots of advice is sufficient to recognize language $\mathtt{L}_m$ for any $m$.

Observe that $\mathtt{L}_m$  has no member of length $n$ if $n$ is not divisible by $m+1$, and it has exactly one member of length $n$ if  $n$ is divisible by $m+1$. A member string is made up of $m+1$ segments of equal length, each of which contains only one symbol. Let us call the positions $ (\frac{n}{m+1}+1),  (\frac{2n}{m+1}+1), \ldots, (\frac{mn}{m+1}+1)$ of a member of $\mathtt{L}_m$, where a new segment starts after the end of the previous one, the ``border points". If $m$ inkdots marking these $m$ border points  are provided as advice, a finite automaton can recognize $\mathtt{L}_m$  by simply accepting a string whose length is a  multiple of $m+1$ if and only if the input consists of alternating runs of 0's and 1's, and that all and only the marked symbols are different than their predecessors in the string.
We have thus proven that 
 $  \mathtt{L}_m\in \mathsf{REG}/m  (\odot)$ for $m  \in \mathbb{Z^+}$.

To show that 
  $\mathtt{L}_m\notin \mathsf{REG}/(m-1)  (\odot)$ for any $m$, suppose that there is a finite automaton $M$ which recognizes the language $\mathtt{L}_m$ with the help of $m -1$ inkdots for some  $m \in \mathbb{Z^+}$. Let $q$ be the number of states of $M$, and let $u$ be an integer greater than $4q^2$. 

Note that the string $s=(0^u1^u)^{\ceil{m/2}}(0^u)^{m+1-2\ceil{m/2}}$ is in $\mathtt{L}_m$, and so it should be accepted by $M$. We will use $s$ to construct another string of the same length which would necessarily be accepted by $M$ with the same advice, although it would  not be a member of $\mathtt{L}_m$.

Note that there are  $m$ border points in $s$. Let us call the $(4q^2+1)$-symbol-long substring of $s$ centered at some border point the \emph{neighborhood} of that border point.  Since $M$ takes at most $m -1$ inkdots as advice,  there should be at least one border point, say, $b$, whose neighborhood contains no inkdot.  Without loss of generality, assume that position $b$ of the string $s$ contains a 1, so the neighborhood is of the form $0^{2q^2}1^{2q^2+1}$. Since $M$ has only $q$ states, and there is no inkdot around, $M$ must ``loop" (\textit{i.e.} enter the same state repeatedly) both while scanning the 0's, and also while scanning the 1's of this neighborhood. Let  $d$  denote the least common multiple of the periods of the two loops (say, $p_1$ and $p_2$) described above. Note that as $p_1 \leq q$ and $p_2 \leq q$,  $d$ can not be greater than $q^2$. 

Now consider the new string $s^\prime$ that is constructed by replacing the aforementioned neighborhood $0^{2q^2}1^{2q^2+1}$ in $s$ with the string $0^{2q^2-d}1^{2q^2+d+1}$. $s^\prime$ is of the same length as $s$, and it is clearly not a 
member of $\mathtt{L}_m$, since it contains segments of different lengths. But $M$ would  nonetheless accept $s^\prime$ with the advice for this input length, since $M$'s computation on $s^\prime$ would be almost the same as that on $s$, with the exception that it loops $d/p_{2}$ more  times on the segment containing position $b$, and $d/p_1$ fewer times on the  preceding segment, which is still long enough (\textit{i.e.} at least $q^2$ symbols long) to keep $M$ looping. $M$ therefore enters the same states when it starts scanning each segment of new symbols during its executions on both $s$ and $s^\prime$, and necessarily ends up in the same state at the end of both computations. This means that  $M$ accepts a non-member of $\mathtt{L}_m$, which contradicts the assumption that $M$ recognizes $\mathtt{L}_m$ with $m - 1$ inkdots as advice.
\end{proof}

It has been shown that  every additional inkdot increases the language recognition power of dfa's that operate with constant amounts of advice. Extending the same argument, we now prove that more and more languages can be recognized if one allows the amount of inkdots given as advice to be faster and faster increasing functions of the length of the input.

\begin{theorem}\label{theorem:inkdot_function_hierarchy}
For every pair of functions $f,g:\mathbb{N}\rightarrow\mathbb{N}$ such that  $f(n) \in \omega(1) \cap o(n)$, if there exists $n_0\in\mathbb{Z}$ such that $g(n) < f(n)-2$ for all $n>n_0$, then $ \mathsf{REG}/ g(n)  (\odot)  \subsetneq \mathsf{REG}/ f(n)  (\odot)$. 
\end{theorem}

\begin{proof}
Let $f$  be a function on $\mathbb{N}$  such that  $f(n) \in \omega(1) \cap o(n)$, and let $f'(n)= \ceil{n/f(n)}$.
 
Note that  $f'(n) \in \omega(1) \cap o(n)$, and $f(n) \geq n/f'(n)$ for all $n$.

Consider the language  $\mathtt{L}_f = \{ w \mid w=w_1 \cdots w_i \cdots w_n \mbox{ where } w_i \in \{0,1\} ~ \mbox{and} ~ w_i= 1 \mbox{ iff } i = k f'(n)$  for some $ k\in Z^+  \}.$

Each member $w$ of $\mathtt{L}_f$ is simply a binary string containing 1's at the $f'(|w|)$'th,  $2f'(|w|)$'th, etc. positions, and $0$'s everywhere else. Since the gaps between the $1$'s is $f'(|w|)$, \textit{i.e.} an increasing function of the input length, sufficiently long members of $\mathtt{L}_f$ can be ``pumped" to obtain nonmembers, meaning that $\mathtt{L}_f$ can not be recognized by a finite automaton without advice. However a finite automaton which takes inkdots as advice can  recognize $\mathtt{L}_f$ easily: The advice for length $n$ consists simply of inkdots placed on the exact positions where $1$'s are supposed to appear in a member string of this length. As a member with length $n$ contains at most $n/f'(n)$ $1$'s, $f(n)$ inkdots are sufficient. We conclude that $\mathtt{L}_f \in  \mathsf{REG}/ f(n)  (\odot)$. 

Now, suppose that $\mathtt{L}_f \in  \mathsf{REG}/ g(n) (\odot)$ for some function $g:\mathbb{N}\rightarrow\mathbb{N}$ such that $g(n) <f(n)-2$. Then there should be a finite automaton $M$ which would recognize $\mathtt{L}_f$ with the help of at most $g(n)$  inkdots.

Note that the precise number of $1$'s in a member of $\mathtt{L}_f$ of length $n$ is given by $f''(n) = \floor{n/f'(n)} =  \floor{n/ \ceil{n/f(n)} }$. For large values of $n$, $f''(n)$ takes values in the set $\{ f(n) -1, f(n)\}$,  since $f(n)\in o(n)$. Therefore $g(n) < f(n)-2$ for sufficiently long inputs implies that  $g(n)$ inkdots will  not be enough to mark all input positions where a member string of that length should contain a $1$. In fact, recalling that the distance between $1$'s in the  members of $\mathtt{L}_f$ of length $n$  is given by $f'(n)$, we see that  sufficiently long members of $\mathtt{L}_f$ must contain at least one 1 which has an inkdot-free neighborhood (in the sense of the term used in the proof of Theorem \ref{theorem:inkdot_hierarchy}), where the all-0 segments to the left and the right of the 1 are long enough to cause $M$ to loop. We can then use an argument identical to the one in that proof to conclude that a non-member of $\mathtt{L}_f$ also has to be accepted by $M$, reaching a contradiction.
\end{proof}

\section{Succinctness issues}
\label{sec:succinctness}

In this section, we  investigate the effects of  advice on the succinctness of finite automata. In particular, we will demonstrate a family of languages where the sizes of the associated minimal automata depend on whether advice is available, and if so, in which format.

 For a given integer $k>1$, we define  
$\mathtt{LANG}_k = \{ w \in \{0,1\}^*   \mid |w| <  k \mbox{ or }  w_{i+1} = 1  \mbox{ where }  i = |w|\mod  k  \}$.

\begin{theorem} 
 A dfa with no advice would need to have at least $2^k$ states in order to recognize $\mathtt{LANG}_k$.
 \end{theorem}
 
  \begin{proof}
Let $x,y$ be any  pair of distinct strings, $x,y \in  \{0,1\}^k $. 
There exists a position $i$ such that $x_i \neq y_i$. Without loss of generality,  assume that $x_i=0$ and $y_i = 1$. Then for any string $z \in \{0,1\}^*$ with $ i -1 = |z|\mod  k$, we have $xz \notin \mathtt{LANG}_k$ and $yz\in \mathtt{LANG}_k$. In other words, the index of $\mathtt{LANG}_k$ (in the sense of the Myhill-Nerode theorem (see \textit{e.g.} (\cite{HU79})) is at least as big as the number of distinct strings in  $\{0,1\}^k$, namely, $2^k$. Therefore, a dfa would need at least that many states in order to recognize $\mathtt{LANG}_k$.  
\end{proof}

Note that testing membership in $\mathtt{LANG}_k$ is as easy as checking whether the $i+1$'st symbol of the input  is a $1$  or a $0$,  where the length $|w|$ of the input word satisfies $|w| = m k +i $. The problem is that in the  setting without the advice, the value $i$ is learned only  after the machine scans the last symbol. $i$, however is a function of the length of the input, and advice can be used to convey this information  to the automaton at an earlier stage of its computation. 

\begin{theorem} 

$\mathtt{LANG}_k$ can be recognized by a ($k+3$)-state dfa with the help of prefix advice. However, no dfa with fewer than $k$ states can recognize $\mathtt{LANG}_k$  with prefix advice.
 
\end{theorem}

\begin{proof}
We describe a ($k+3$)-state machine $M_1$, which takes binary prefix advice of length $k$ to recognize $\mathtt{LANG}_k$. For inputs of length less than $k$, the advice is $0^k$. For longer inputs of length $i$ (mod $k$), the advice is the string $0^{i}10^{k-i-1}$. 

$M_1$ is depicted in Figure 1, where double circles are used to indicate accepting states.
$M_1$ remains at its initial state, which is an accepting state, as long as it scans $0$'s on the combined advice-input string. If it scans a $1$, it attempts to verify if the $k$'th symbol after that $1$ is also a $1$. If the input ends before that position is reached, $M_1$ accepts. If the input is sufficiently long to allow that position to be reached, $M_1$ accepts if it scans a $1$ there, and rejects otherwise. It is easy to see that $k+3$ states are sufficient to implement $M_1$.

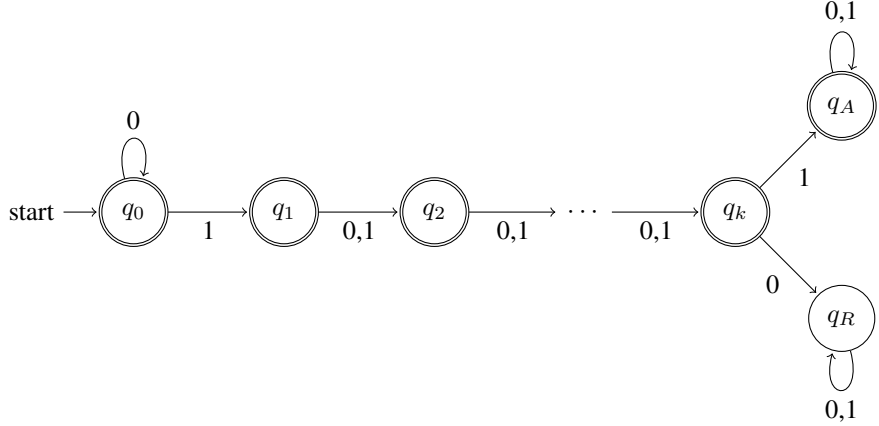
\begin{figure}\label{fig:m1}
\begin{center}
\begin{tikzpicture}[shorten >=1pt,node distance=2cm,on grid,auto] 
   \node[state,initial,accepting] (q_0)   {$q_0$}; 
   \node[state,accepting] (q_1) [ right=of q_0] {$q_1$}; 
   \node[state,accepting] (q_2) [ right=of q_1] {$q_2$}; 
   \node        (q_dots) [right=of q_2] {$\cdots$}; 

   \node[state,accepting] (q_k) [ right=of q_dots] {$q_k$}; 
   \node[state,accepting] (q_A) [above right=of q_k] {$q_A$}; 
   \node[state] (q_R) [below right=of q_k] {$q_R$};

    \path[->] 
    (q_0) edge [loop above] node {0} ()
	   edge  node [swap] {1} (q_1)
    (q_1) edge node [swap] {0,1} (q_2)
    (q_2) edge node [swap] {0,1} (q_dots)
    (q_dots) edge node [swap] {0,1} (q_k)
    (q_k) edge node [swap] {1} (q_A)
    (q_k) edge node [swap] {0} (q_R)
    (q_A) edge [loop above] node {0,1} ()
    (q_R) edge [loop below] node {0,1} ()

;

\end{tikzpicture}
\caption{Finite automaton $M_1$ with k+3 states recognizes $\mathtt{LANG}_k$ with prefix advice.}
\end{center}
\end{figure}

Let us now assume that a ($k-1$)-state finite automaton $M_2$ recognizes $\mathtt{LANG}_k$ with the help of prefix advice. Then $M_2$ must accept the string $s = 0^{k-1}10^{k^2 + k-1}$, which is a member of $\mathtt{LANG}_k$, utilizing the advice for ($k^2+2k-1$)-symbol-long inputs. Since the $0$ sequences  to the left and right of the single $1$ in $s$ are of length at least $k-1$,  $M_2$  must ``loop" during its traversals of these sequences while scanning $s$. Let $p<k$ and $q<k$ be the periods of these loops to the left and  right of $1$, respectively.  Now consider the string $s' = 0^{k-1 + pq}10^{k^2 + k-1 - pq}$. Note that $s'$ is of the same length as $s$, and therefore it is assigned the same advice as $s$. Also note that $M_2$ must enter the same states at the ends of the $0$ sequences on both $s'$ and $s$, since  $q$ additional iterations of the loop to the left of the $1$ and $p$ fewer iterations of the loop to the right of the $1$ does not change the final states reached at the end of these $0$ sequences. (The ``pumped down'' version of the zero sequence to the right of $1$ is $k^2 + k-1 - pq$ symbols long and since we have $pq < k^2$ this means it is still sufficiently long (\textit{i.e.} longer than $k$ symbols) to ensure that  it causes at least one iteration of the loop.) This implies that $M_2$ should accept $s'$ as well, which is a contradiction, since $s'$ is not a member of $\mathtt{LANG}_k$. 
\end{proof}

\begin{theorem} 
  There exists a dfa with just two states that can recognize $\mathtt{LANG}_k$ with the help of inkdot advice for any $k$.
\end{theorem}

\begin{proof}
For inputs of length less than $k$, no inkdots are provided as advice. For longer inputs whose length is $i$ mod $k$, the advice is an inkdot on the $i+1$'st symbol of the input. A dfa that accepts if and only if either senses no inkdots or if the single marked position contains a 1 can be constructed using two states, as seen in Figure 2.
\end{proof}

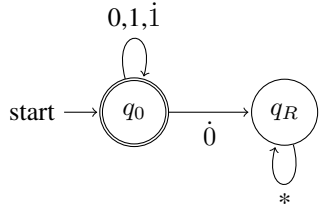
\begin{figure}
\begin{center}
\begin{tikzpicture}[shorten >=1pt,node distance=2cm,on grid,auto] 
   \node[state,initial,accepting] (q_0)   {$q_0$}; 
   \node[state] (q_R) [ right=of q_0] {$q_R$};

    \path[->] 
    (q_0) edge [loop above] node {0,1,$\dot{1}$ } ()
    (q_0) edge node [swap] {$\dot{0}$  }    (q_R)
    (q_R) edge [loop below] node {*} ();

\end{tikzpicture}
\caption{Finite automaton with 2 states recognizes $\mathtt{LANG}_k$ with inkdot advice.}
\end{center}
\end{figure}

\section{Randomized inkdot advice to finite automata}
\label{sec:random}

It is known that ``randomized" advice, selected for each input from a set of alternatives according to a particular distribution, is even more powerful than its deterministic counterpart in several models (see (\cite{Ya10}) and (\cite{KSY14})). In this section, we will show that a similar increase in power is also the case for the inkdot advice model, even with a restriction to a constant number of inkdots.

For each input length $n$, the advice is a set of pairs of the form $\langle I, p\rangle$, where $I$ is an inkdot pattern to be painted on an $n$-symbol-long string, and $p$ is the probability with which pattern $I$ is to be used, with the condition that the probabilities add up to 1. Whenever the advised dfa is presented with an input, an inkdot pattern is selected from the advice with the corresponding probability. A language is said to be \emph{recognized with bounded error} by such a machine if the automaton responds to each input string with the correct accept/reject response with probability at least $\frac{2}{3}$. The class of languages recognized with bounded error by a dfa aided with randomized advice in the advice track model is called $ \mathsf{REG}/ Rn $ in (\cite{Ya10}). We call the corresponding classes for $k$ inkdots $ \mathsf{REG}/ Rk(\odot) $ for $k \in \mathbb{Z^+}$, and $ \mathsf{REG}/ Rn(\odot) $ when there are no restrictions on the number of inkdots.

By importing Proposition 16 in (\cite{Ya10}) about the advice track model  directly to the inkdot advice case, one sees immediately that $ \mathsf{REG}/ n(\odot)\subsetneq  \mathsf{REG}/ Rn(\odot) $. We will present a new result, showing that randomization adds power even when the number of inkdots is restricted to be a constant.

\begin{theorem}\label{theorem:randomized_advice}
$ \mathsf{REG}/ 2(\odot)  \subsetneq \mathsf{REG}/ R2 (\odot)$. 
\end{theorem}

\begin{proof}
As deterministic advice is a special case of the probabilistic version, the inclusion is trivial.
In order to show the separation, consider the language $\mathtt{L}_3=\{0^m1^m0^m1^m \mid m>0\}$, defined as part of a family in  Section \ref{sec:moredots}. By the proof of Theorem \ref{theorem:inkdot_hierarchy}, $\mathtt{L_3}\notin \mathsf{REG}/2  (\odot)$. It remains to show that
 $\mathtt{L_3}\in \mathsf{REG}/R2  (\odot)$.

Recall from the proof of Theorem \ref{theorem:inkdot_hierarchy} that we call the $m+1$'st, $2m+1$'st, and $3m+1$'st positions of the string $0^m1^m0^m1^m$ the ``border points".  We build a dfa that will take two inkdots as advice to recognize $\mathtt{L_3}$. In the randomized advice that will be given for inputs of length $4m$, the three pairs of border points $(m+1, 2m+1)$, $(m+1, 3m+1)$, and $(2m+1, 3m+1)$ will be marked with probability $\frac{1}{3}$ each. The dfa simply checks whether its input is of the form $0^+1^+0^+1^+$, the input length is a multiple of 4, and also whether each inkdot that it sees is indeed on a symbol unequal to the previously scanned symbol. The machine accepts if and only if all these checks are successful.  

If the input string is a member of $\mathtt{L}_3$, all checks will be successful, and the machine will accept with probability 1. For a nonmember string $s$ of length $4m$ to be erroneously accepted, $s$ must be of the form $0^+1^+0^+1^+$, and contain only one ``false border point," where a new segment of 1's (or 0's) starts at an unmarked position after a segment of 0's (or 1's). But this can happen with probability at most $\frac{1}{3}$, since two of the three possible inkdot patterns for this input length must mark a position which contains a symbol that equals the symbol to its left. 
\end{proof}

\section{Advised computation with arbitrarily small space}
\label{sec:sublogspace}


In this section, we will consider one way Turing machines, instead of finite automata, as the underlying advised machine model in order to explore the effects of combining inkdot advice with non-constant, but still very small amounts of memory. 

It is known that unadvised deterministic Turing machines with sublogarithmic space which scan their input once from left to right can only recognize regular languages (see (\cite{SHL65})). Therefore such small amounts of additional workspace does not increase the computational power of dfa's.  On the other hand, sublogarithmic space can be fruitful  for nonuniform computation. In (\cite{RCH91}), for example, $\log \log n$ space was proven to be necessary and sufficient for a demon machine, which is defined as a Turing machine whose worktape is limited with endmarkers to a prespecified size determined by the length of the input,  to recognize the nonregular language $\{a^nb^n \mid n > 1\}$. 

Deterministic machines, ($\log \log n$ space Turing machines in particular) was shown in (\cite{RCH91}) not to gain any additional computational power if they are also provided with the ability of marking one input tape cell with an inkdot.



Below, we will show that a one-way Turing machine which has simultaneous access to arbitrarily slowly increasing amounts of space and one inkdot provided as advice can effectively use both of these resources in order to extend its power of language recognition.  Note that the head on the single work tape of the TM is allowed to move in both directions.
  
We extend our notation for the purposes of this section as follows. The class of languages recognized by one-way-input Turing machines which use $s(n)$ cells in their work tape when presented inputs of length $n$ will be denoted by $\mathsf{1SPACE}(s(n))$. With an advice of $k$ inkdots, the corresponding class is called $\mathsf{1SPACE}(s(n))/k(\odot)$.

\begin{theorem} \label{th:sublogspace}
For any slowly increasing function  $g(n):\mathbb{Z}^+ \rightarrow \mathbb{Z}^+$, where $g(n) \in \omega(1) \cap o(n)$, 
\begin{itemize}
\item[-] $ \mathsf{1SPACE}(\log g(n)) \subsetneq \mathsf{1SPACE}(\log  g(n))/1  (\odot)$,
\item[-] $ \mathsf{REG}/1  (\odot) \subsetneq \mathsf{1SPACE}(\log g(n))/1  (\odot)$. 
\end{itemize}

In other words, automata with access to arbitrarily slowly increasing amounts of space and a single inkdot as advice are more powerful than those which lack either of these resources.
\end{theorem}

\begin{proof}
Both inclusions are straightforward, and we proceed to  show the separation results. First, note that $ \mathsf{1SPACE}(\log g(n)) \subseteq \mathsf{REG}/k(\odot) $ for any such $g$ and any $k$, since     one-way deterministic Turing machines can not utilize sublogarithmic space in order to recognize nonregular languages (see (\cite{SHL65})), and are therefore computationally equivalent to dfa's without advice. It will therefore be sufficient to demonstrate a language which is  in $\mathsf{1SPACE}(\log  g(n))/1  (\odot)$ but not in  $ \mathsf{REG}/1  (\odot)$. For this purpose, we define the language $\mathtt{L}_{g}$ over the alphabet $\{ \#, 0,1\}$ to be the collection of the strings of the form $ s_1\#s_2\# \cdots \#s_m\#^+$,  where 
\begin{itemize}
\item[-] for a member of length $n$ and for $i = 1,2, \ldots, m $;  $s_i$ is a subword of the  form $0^*10^*$, the  length of which is given by $\floor{g(n)}$,  (meaning that $m$ is at most  $\floor{n/(\floor{g(n)}+1)}$), and,
\item[-]  For all $s_i$, $p_i$ denotes the position of the  symbol $1$  within that subword, and  $p_i \in \{ p_{i-1} -1 , p_{i-1}, p_{i-1} +1 \}$ for $i = 2,3, \ldots ,m$.

\end{itemize}

As mentioned earlier, Fact \ref{yamakami_characterization} provides a useful tool for showing that certain languages are not in $ \mathsf{REG}/n(\odot)$. We see that $L_g$  is in $ \mathsf{REG}/n$ if and only if the number of equivalence classes of the equivalence relation $\equiv_{L_g}$ associated with $L_g$ in the way described in Fact 1 is finite. 

Considering pairs of distinct strings $x$ and $y$ of the form $0^*10^*$ such that  $|x| = |y|=\floor{g(n)}$, for values of  $n$ which are big enough that $g(n)\ll n$, one sees that $\equiv_{L_g}$ must have at least $\floor{g(n)}$ equivalence classes. This is because one has to distinguish $\floor{g(n)}$ different subword patterns to decide whether the relationship dictated between subwords $s_1$ and $s_2$ holds or not. Noting that the number of equivalence classes of $\equiv_{L_g}$ is not finite, we conclude that $\mathtt{L}_g$ is not even in $ \mathsf{REG}/n$, let alone in $ \mathsf{REG}/n(\odot)$.

To prove $\mathtt{L}_g \in \mathsf{1SPACE}(\log  g(n))/1  (\odot)$, we describe a one-way Turing machine $T$ that uses one inkdot and $O(\log  g(n))$ space to recognize $\mathtt{L}_g$. The advice for strings of length $n$ is a single inkdot on the $\floor{g(n)}$'th position, \textit{i.e.} where the last symbol of the first subword $s_1$ should appear in a member of the language. During its left-to-right scan of the input, $T$ performs the following tasks in parallel:

\begin{itemize}
\item[-] $T$  checks if its input is of the form $(0^*10^*\#)^+\#^*$, rejecting if this check fails.

\item[-] Using a counter on its work tape, $T$ counts the rightward moves of the input head up to the point where the inkdot is scanned. Having thus ``learned" the value $\floor{g(n)}$, $T$ notes the number of bits required to write this value on the work tape, and marks the tape so as to never use more than a fixed multiple of this number of cells. It compares the lengths of all subsequent subwords with this value, and rejects if it sees any subword of a different length. $T$ also rejects if the cell with the inkdot does not contain the last symbol of the first subword. 

\item[-] $T$ checks the condition $p_i \in \{ p_{i-1} -1 , p_{i-1}, p_{i-1} +1 \}$ for $i = 2,3,\ldots ,m$, by using another pair of counters to record the position of the  symbol 1 in each subword, rejecting if it detects a violation. 
\item[-] $T$ accepts if it finishes scanning the input without a rejection by the threads described above.

\end{itemize}
Clearly, the amount of memory used by $T$ is just a fixed multiple of $\log g(n)$.

\end{proof}


Let us also note that the argument of the proof above can be applied to the general advice track and advice tape models, showing that such small amounts of space are useful for advised automata of those types as well.

\section{Conclusion}
\label{sec:conc}

This paper introduced inkdots as a means of providing advice to finite automata. Nontrivial results on the power of this model and its relationships with other models of advised automata were presented in various settings.

We conclude with a list of open questions.
\begin{itemize}
\item Is there a language that can be recognized by a dfa with $\floor{n/2}$, but not $\floor{n/2}-1$ inkdots of advice?
\item Are there hierarchies like those in Theorems \ref{theorem:inkdot_hierarchy} and \ref{theorem:inkdot_function_hierarchy} for randomized advice as well?
\item We considered only deterministic models for our advised machines. How would things change if we allowed probabilistic or quantum models?
\item Can the inkdot model be extended to other machine models, like pushdown automata? 
\item Are there other, even ``weaker" types of advice? 
\end{itemize}

\acknowledgements
\label{sec:ack}
 We thank Alper \c{S}en, Atilla Y{\i}lmaz, and the anonymous reviewers for their helpful comments.

\bibliographystyle{abbrvnat}
\bibliography{inkdots-bib-dmtcs}
\label{sec:biblio}

\end{document}